\documentclass[runningheads]{llncs}
\usepackage{graphicx}
\usepackage{amsmath}
\usepackage{amsfonts}
\usepackage{enumerate}
\usepackage{verbatim}
\usepackage{algebra}
\begin{document}
\title{An algorithmic approach to the existence of ideal objects in commutative algebra
\thanks{The first, second and third author were supported by the German Science Foundation (DFG Project KO 1737/6-1); 
by the John Templeton Foundation (ID 60842) and 
by a Marie Sk\l odowska-Curie fellowship of the Istituto Nazionale di Alta Matematica, respectively. 
{The opinions
expressed in this paper are those of the authors and do not necessarily reflect the views of the John Templeton Foundation.}
}
}
\titlerunning{An algorithmic approach to the existence of ideal objects}
%
\author{Thomas Powell
\inst{1}
\and
Peter Schuster
\inst{2}
\and
Franziskus Wiesnet
\inst{3}
}
\authorrunning{Powell et al.}
%
\institute{Technische Universit\"{a}t Darmstadt \and
University of Verona \and University of Trento} 

%
\maketitle              
\begin{abstract}
The existence of ideal objects, such as maximal ideals in
nonzero rings, plays a crucial role in commutative algebra. These are
typically justified using Zorn's lemma, and thus pose a challenge from a
computational point of view. Giving a constructive meaning to ideal
objects is a problem which dates back to Hilbert's program, and today is still
a central theme in the area of dynamical algebra, which focuses on
the elimination of ideal objects via syntactic methods. In this paper, we
take an alternative approach based on Kreisel's no counterexample interpretation and
sequential algorithms. We first give a computational interpretation to an
abstract maximality principle in the countable setting via an intuitive,
state based algorithm. We then carry out a concrete case study, in which
we give an algorithmic account of the result that in any commutative ring,
the intersection of all prime ideals is contained in its nilradical.

\end{abstract}

\section{Introduction}

This paper is an application of proof theory in commutative algebra. To be more precise, we use proof theoretic methods to give a computational interpretation 
to a general maximality principle (Theorem \ref{thm-zorn}), which in particular implies the existence of maximal ideals in commutative rings (Krull's lemma). 
In the context of second order arithmetic, the latter statement is equivalent to arithmetical comprehension \cite[Chapter III.5]{Simpson(1999.0)}, 
and thus Theorem \ref{thm-zorn} is a genuinely strong principle, and highly non-trivial from a computational perspective.

The extraction of programs from proofs has a long and rich history, dating back to Kreisel's pioneering work on the `unwinding' of proofs \cite{Kreisel(1951.0),Kreisel(1952.0)}. In the ensuing decades, the application of proof interpretations in particular has become a major topic in proof theory, and today encompasses both \emph{proof mining} \cite{Kohlenbach(2005.0),Kohlenbach(2008.0),Kohlenbach(2019.0)}, which focuses on obtaining quantitative information primarily from proofs in areas of mathematical analysis, and the mechanized synthesis of programs from proofs, which has found many concrete applications in discrete mathematics and computer science \cite{BergLFS(2015.0),BergMSS(2011.0),SchSeiWies(2015.0)}. 

Though as far back as the 1950s Kreisel already discusses the use of proof theoretic techniques to extract quantitative information from proofs in abstract algebra \cite{Kreisel(1958.0)}, specifically Hilbert's 17th problem together with his \emph{Nullstellensatz}, to date there are comparatively few formal applications of proof interpretations in algebra, 
the computational analysis of which is done largely on a case by case basis. 
This typically involves replacing \emph{semantic} conservation theorems with appropriate \emph{syntactic} counterparts both sufficient for proofs of elementary statements and provable by elementary means. 
This method has proved possible in numerous different settings \cite{cedcoq:entaildistr,ced:heine,mulvey:globhb,neg:ord,rinwes:sik,wes:ogc}, 
and in the context of \emph{commutative} algebra the so-called dynamical method is especially dominant 
\cite{cos:dyn,lombardiquitte:constructive,yengui:maximal,yengui:constructive}.\footnote{ 
The second author has contributed to a universal conservation criterion
\cite{rin:edde,rin:edde:full,rin:cuts} that includes many of the those cases \cite{rin:ukl,sch:hah,wes:ext}.} 
In dynamical algebra one deals with a supposed ideal object (such as a maximal ideal) 
only by means of concrete, finitary approximations (such as finitely generated ideals, or rather the finite sets of generators), 
where the latter provide partial but sufficiently complete information about the former. 



Interestingly, the idea of replacing ideal objects with suitable finitary approximations is already implicit in Kreisel's unwinding program, and is captured by his famous no-counterexample interpretation (n.c.i.). The n.c.i. plays an important role in proof mining, where in particular it corresponds to the notion of \emph{metastability} \cite{Kohlenbach(2005.1),KohKou(2015.0),KohLeu(2012.0)}, which has been made popular by Tao \cite{Tao(2008.1)} and more recently has featured in higher order computability theory \cite{Sanders(2018.0)}. 

In this article, we take a new approach to eliminating ideal objects in abstract algebra, by solving an appropriate metastable reformulation of our general maximality principle. We then use this solution to extract \emph{direct witnesses} from a variant of Krull's lemma.

The novelty of our approach lies not just in our use of the n.c.i., but in our description of its solution as a state based algorithm, 
inspired by recent work of the first author \cite{OliPow(2015.0),OliPow(2017.0),Powell(2013.0),Powell(2016.0),Powell(2018.3)} 
which focuses on the algorithmic meaning of extracted programs. This form of presentation allows us to bridge the gap between the \emph{rigorous} 
extraction of programs from proofs as terms in some formal calculus, and the more algorithmic style of dynamical algebra. 

It also enables us to present our results in an entirely self-contained manner, without needing to introduce any heavy proof theoretic machinery. Though behind the scenes at least, aspects of our work are influenced by G\"{o}del's functional interpretation \cite{Goedel(1958.0)} and Spector's bar recursion \cite{Spector(1962.0)}, neither of these make an official appearance, and we have endeavoured to make everything as accessible to the non-specialist as possible.

Our first main contribution, given as Theorems \ref{thm-verify} and \ref{thm-termination}, is a time sequential algorithm (in the sense of Gurevich \cite{Gurevich(2000.0)}), whose states evolve step by step until they terminate in some final state $s_j$ which represents a solution to the n.c.i. of Theorem \ref{thm-zorn}. Each step in this process represents an \emph{improvement} to our construction of an approximate ideal object, and so can also be viewed as a learning procedure in the style of \cite{AscBer(2010.0)}.

We then present a concrete application of our abstract result, in which we analyse a classic maximality argument used to prove the well known fact 
that in any commutative ring, if some element $r$ is contained in intersection of all prime ideals, then it must be nilpotent. We show that an instance 
of our sequential algorithm can be used to directly compute an exponent $e>0$ such that $r^e=0$, and thus our case study is another illustration of how the proof 
theoretic analysis of a highly nonconstructive proof can yield direct, computational information. We conclude by instantiating our algorithm in case of nonconstant coefficients of invertible polynomials. This is a well known example which has been widely studied 
from a computational perspective \cite{per:spe,rich:trivial,LICS2012,LMCS2013}, thus facilitating a future analysis of our work with other approaches.


\section{A general maximality argument}

We begin by presenting our abstract maximality principle, which forms the main subject of the paper. 
Let $X$ be some set (which for now is arbitrary but later will be countable), and denote by $\psfin(X)$ 
the set of all finite subsets of $X$. Simple lemmas are stated without proof.

\begin{definition}
\label{def-gen}
Let $\rhd$ be some subset of $\psfin(X)\times X$. We treat $\rhd$ as a binary relation and say that the element $x$ is generated by the finite set $A$ whenever $A\rhd x$. We extend $\rhd$ to arbitrary (not necessarily finite) $S\subseteq X$ by defining $S\rhds x$ whenever there exists some finite $A\subseteq S$ such that $A\rhd x$.
\end{definition}

\begin{definition}
\label{def-closure}
Given some $S\subseteq X$, define the sequence $(S_i)_{i\in\NN}$ of sets by
\begin{equation*}
S_0:=S\mbox{ \ \ \ and \ \ \ }S_{i+1}:=\{x\; | \; \bigcup_{j\leq i}S_j\rhds x\}
\end{equation*}
and let $\closure{S}:=\bigcup_{i\in\NN}S_i$. We say that $\closure{S}$ is the closure of $S$ w.r.t. $\rhd$, since whenever $\closure{S}\rhds x$ then $x\in \closure{S}$. 
\end{definition}



\begin{definition}
\label{def-plus}
For any $S\subseteq X$ and $x\in X$, $S\oplus x:=\closure{S\cup\{x\}}$ denotes the closed extension of $S$ with $x$.
\end{definition}

\begin{lemma}
\label{lem-plus}
Suppose that $S\rhds x$. Then $S\oplus x=\closure{S}$.
\end{lemma}

\begin{definition}
Let $Q(x)$ be some predicate on $X$. For $S\subseteq X$ write $Q(S)$ for $(\forall x\in S) Q(x)$. Note in particular that $Q(S)$ and $S\supseteq T$ implies $Q(T)$.
\end{definition}

\begin{definition}
\label{def-maximal}
We say that $M\subseteq X$ is maximal w.r.t. $\rhd$ and $Q$ if 
\begin{enumerate}[(i)]

\item $M$ is closed w.r.t. $\rhds$,

\item $Q(M)$,

\item $\neg Q(M\oplus x)$ for any $x\notin M$.

\end{enumerate}
   
\end{definition}

\begin{theorem}
\label{thm-zorn}
Suppose that $Q(\closure{\emptyset})$. Then there exists some $M\subseteq X$ which is maximal w.r.t. $\rhd$ and $Q$.
\end{theorem}

\begin{proof}
Define $\mathcal{S}:=\{S\subseteq X\; | \; \mbox{$S$ is closed w.r.t $\rhd^\ast$ and $Q(S)$}\}$. We show that $\mathcal{S}$ is nonempty and chain complete w.r.t. set inclusion. Nonemptyness follows from the fact that $\closure{\emptyset}\in \mathcal{S}$, so it remains to prove chain completeness. Let $S_0\subseteq S_1\subseteq\ldots$ be a chain in $\mathcal{S}$. Then $S:=\bigcup_{i\in\NN} S_i$ is clearly closed, and moreover, if $x\in S$ then $x\in S_j$ for some $j$, and therefore $Q(x)$. This establishes $S\in\mathcal{S}$.

Thus by Zorn's lemma, $\mathcal{S}$ has some maximal element $M$, which by definition satisfies (i) and (ii). But for $x\notin M$ we have $M\subset M\oplus x$ and thus $M\oplus x\notin \mathcal{S}$. But since $M\oplus x$ is closed, it follows that $\neg Q(M\oplus x)$.
\end{proof}

\begin{corollary}
\label{cor-ideal}
Any commutative ring $X$ with $0\neq 1$ has a maximal ideal.
\end{corollary}

\begin{proof}
We follow the standard proof. Define $\rhd$ by $A\rhd x$ iff $x=x_1\cdot a_1+\ldots +x_k\cdot a_k$ for some $a_1,\ldots,a_k\in A$ and $x_1,\ldots,x_k\in X$. In addition, define $Q(x):\equiv (x\neq 1)$. Then $S\subseteq X$ is closed iff it is an ideal, with $Q(S)$ iff $S$ is proper. Now $\closure{\emptyset}=\{0\}$ (since $\emptyset\rhd 0$) and if $0\neq 1$ then $Q(\{0\})$, thus by Theorem \ref{thm-zorn} there exists some maximal structure $M$. To see that $M$ is a maximal ideal, if there were some $M\subset I\subseteq X$ then we would have $M\subset M\oplus x\subseteq I$ for some $x\notin M$, and by $\neg Q(M\oplus x)$ we would have $1\in M\oplus x$ and thus $I=X$.
\end{proof}

\section{A logical analysis of Theorem \ref{thm-zorn}}

From now on, we assume that $X$ is countable and comes equipped with some explicit enumeration $\{x_n\; | \; n\in\NN\}$. Given some $S\subseteq X$, the initial segment of $S$ of length $n$ is defined by $\initSeg{S}{n}:=S\cap\{x_m\; | \; m<n\}$. Note that $S=\bigcup_{n\in\NN}\initSeg{S}{n}$. We define $\dom(S)\subseteq \NN$ by $\dom(S):=\{n\in\NN\; | \; x_n\in S\}$.

\begin{theorem}
\label{thm-dc}
Suppose that $M\subseteq X$ satisfies
\begin{equation}
\label{eqn-dcideal}
x_n\in M\Leftrightarrow Q(\initSeg{M}{n}\oplus x_n)
\end{equation}
for all $n\in\NN$. If $Q(\closure{\emptyset})$ then $M$ is maximal w.r.t. $\rhd$ and $Q$. 
\end{theorem}

\begin{proof}
Let $M_n:=\closure{\initSeg{M}{n}}$. We first observe that $Q(M_n)$ for all $n\in\NN$, which follows by induction: For $n=0$ we have $M_0=\closure{\emptyset}$ and so $Q(M_0)$ is true by assumption. Now supposing that $Q(M_n)$ holds for some $n\in\NN$ there are two possibilities: If $Q(\initSeg{M}{n}\oplus x_n)$ then $x_n\in M$ and hence $M_{n+1}=\closure{\initSeg{M}{n}\cup\{x_n\}}=\initSeg{M}{n}\oplus x_n$, and if $\neg Q(\initSeg{M}{n}\oplus x_n)$ then $x_n\notin M$ and hence $M_{n+1}=\closure{\initSeg{M}{n}}=M_n$. Either way we have $Q(M_{n+1})$.

We now establish each of the maximality conditions in turn. For closure, suppose that $M\rhds x_n$ but $x_n\notin M$, and so by definition $\neg Q(\initSeg{M}{n}\oplus x_n)$. Since $M\rhds x_n$ we have $\initSeg{M}{k}\rhds x_n$ for some $k\in\NN$. 
First, let $k\leq n$. Then $\initSeg{M}{k}\subseteq \initSeg{M}{n}$ and thus $\initSeg{M}{n}\rhds x_n$, which implies that $x_n\in M_n$ and thus by Lemma \ref{lem-plus} 
\begin{equation*}
\initSeg{M}{n}\oplus x_n=\closure{\initSeg{M}{n}}=M_n. 
\end{equation*}
Since $Q(M_n)$ this contradicts $\neg Q(\initSeg{M}{n}\oplus x_n)$. 
But if $n<k$ then $\initSeg{M}{n}\oplus x_n\subseteq \initSeg{M}{k}\oplus x_n$ and thus $\neg Q(\initSeg{M}{n}\oplus x_n)$ implies $\neg Q(\initSeg{M}{k}\oplus x_n)$. But $\initSeg{M}{k}\rhds x_n$ and thus by Lemma \ref{lem-plus} again, $\initSeg{M}{k}\oplus x_n=M_k$, contradicting $Q(M_k)$.

That $Q(M)$ holds is straightforward: For if $x_n\in M$ then $x_n\in \initSeg{M}{n+1}\subseteq M_{n+1}$ and thus $Q(x_n)$ follows from $Q(M_{n+1})$. Finally, to show that $\neg Q(M\oplus x_n)$ for $x_n\notin M$, note that $x_n\notin M$ implies $\neg Q(\initSeg{M}{n}\oplus x_n)$, and since $\initSeg{M}{n}\oplus x_n\subseteq M\oplus x_n$ the result follows.
\end{proof}

The purpose of the above theorem was to give a more syntactic formulation of Theorem \ref{thm-zorn} in the countable setting: If $Q(\closure{\emptyset})$ then the existence of a some maximal $M\subseteq X$ is implied by the existence of some $M$ satisfying (\ref{eqn-dcideal}). In order to proceed, we will now take a 
closer look at the structure of (\ref{eqn-dcideal}) and make some restrictions on the logical complexity of certain parameters.

\begin{lemma}
\label{lem-tree}
Suppose that the relation $A\rhd x$ can be encoded as a $\Sigma^0_1$-formula. Then the membership relation $x\in\closure{A}$ can also be encoded as a $\Sigma^0_1$-formula.
\end{lemma}

\begin{proof}
We have $x\in \closure{A}$ iff there exists some finite derivation tree for $x$ whose leaves are elements of $A$ and whose nodes represent instances of $\rhd$. Given that $\rhd$ can be encoded as a $\Sigma^0_1$-formula, it is clear that the existence of a derivation trees can in turn be represented as $\Sigma^0_1$-formula via a suitable encoding. 
\end{proof}

\begin{lemma}
\label{lem-logcomp}
Suppose that $Q(x)$ is a $\Pi^0_1$-formula and that $A\rhd x$ can be encoded as a $\Sigma^0_1$-formula. Then $Q(\closure{A})$ is a $\Pi^0_1$-formula i.e. $Q(\closure{A})\Leftrightarrow (\forall p) R_A(p)$ for some decidable predicate $R_A(p)$ on $\psfin(A)\times\NN$.
\end{lemma}

\begin{proof}
We can write $Q(x)\Leftrightarrow (\forall e) Q_0(x,e)$ for some decidable $Q_0(x,e)$, and by Lemma \ref{lem-tree}, $x\in\closure{A}\Leftrightarrow (\exists t) G_A(x,t)$ for some decidable $G_A(x,t)$. Then 
\begin{equation*}
\begin{aligned}
Q(\closure{A})&\Leftrightarrow (\forall m)(x_m\in\closure{A}\Rightarrow Q(x_m))\\
&\Leftrightarrow (\forall m)((\exists t) G_A(x_m,t)\Rightarrow (\forall e) Q_0(x_m,e))\\
&\Leftrightarrow (\forall m,t,e)(G_A(x_m,t)\Rightarrow Q_0(x_m,e))
\end{aligned}
\end{equation*}
and the latter formula can be encoded as $(\forall p) R_A(p)$ for suitable $R_A(p)$ and using some pairing function for the tuple $m,t,e$.
\end{proof}

\begin{lemma}
Under the conditions of Lemma \ref{lem-logcomp}, (\ref{eqn-dcideal}) holds iff for all $n\in\NN$:
\begin{equation}
\label{eqn-dccomp}
x_n\in M\Leftrightarrow (\forall p)R_{\initSeg{M}{n}\cup\{x_n\}}(p)
\end{equation}
\end{lemma}

\begin{proof}
By Lemma \ref{lem-logcomp} setting $A=\initSeg{M}{n}\cup\{x_n\}$, so that $\closure{A}=\initSeg{M}{n}\oplus x_n$.
\end{proof}

Written out in full, the existence of some $M$ satisfying (\ref{eqn-dccomp}) becomes
\begin{equation*}
(\exists M)(\forall n)(x_n\in M\Rightarrow (\forall p) R_{\initSeg{M}{n}\cup\{x_n\}}(p) \wedge x_n\notin M\Rightarrow (\exists q) R_{\initSeg{M}{n}\cup\{x_n\}}(q))
\end{equation*}
and so written out in Skolem normal form, this becomes
\begin{equation}
\label{eqn-snf}
(\exists M,f)(\forall n,p)(x_n\in M\Rightarrow  R_{\initSeg{M}{n}\cup\{x_n\}}(p) \wedge x_n\notin M\Rightarrow  R_{\initSeg{M}{n}\cup\{x_n\}}(f(n))).
\end{equation}
This motivates our final version of maximality, which is now in a form where we can directly apply the no-counterexample interpretation.

\begin{definition}
\label{def-expmax}
An explicit maximal object w.r.t. $\rhd$ and $Q$ is a set $M\subseteq X$ together with a function $f:\dom(X\backslash M)\to \NN$ such that
\begin{itemize}

\item $x_n\in M\Rightarrow R_{\initSeg{M}{n}\cup\{x_n\}}(p)$

\item $x_n\notin M\Rightarrow \neg R_{\initSeg{M}{n}\cup\{x_n\}}(f(n))$

\end{itemize}
for all $n,p\in\NN$.
\end{definition}
The idea here is that the function $f$ provides concrete evidence for why $x_n$ is excluded from the maximal structure $M$: in other words, it encodes an element $x_{m}$ together with some tree $t$ and $e$ such that $x_{m}\in \initSeg{M}{n}\oplus x_n$ with respect to $t$ but $Q(x_{m})$ fails relative to $e$.

\section{An approximating algorithm for maximal objects}

In general, it is impossible to effectively compute a set $M$ together with an $f$ satisfying Definition \ref{def-expmax}. However, we demonstrate how an \emph{approximate}, or \emph{metastable}, formulation of maximality in the spirit of Kreisel's no-counterexample interpretation, can be directly witnessed via an intuitive stateful procedure.

For a detailed and modern account of the n.c.i., the reader is encouraged to consult e.g. \cite{Kohlenbach(1999.0),Kohlenbach(2008.0)}. The rough idea is the following: Given some prenex formula of the form $A:\equiv (\exists x\in X)(\forall y\in Y) P_0(x,y)$, a functional $\Phi:(X\to Y)\to X$ is said to witness the n.c.i. of $A$ if it witnesses $(\forall\omega:X\to Y)(\exists x) P_0(x,\omega(x))$ i.e. $(\forall \omega) P_0(\Phi \omega,\omega(\Phi \omega))$. This definition generalises in the obvious way to prenex formulas of arbitrary complexity. In this section, we give an algorithmic description of such an $\Phi$ for $A$ being the statement that an explicit maximal object exists, as in Definition \ref{def-expmax}.

\begin{definition}
\label{def-expappr}
Let $(\omega,\phi)$ be functionals which take as input $M$ and $f$ and return as output a tuple in $\NN^2$. An \emph{approximate} explicit maximal object w.r.t $\rhd$, $Q$ and $(\omega,\phi)$ is a set $M\subseteq X$ together with a function $f$ such that 
\begin{itemize}

\item $x_n\in M\Rightarrow R_{\initSeg{M}{n}\cup\{x_n\}}(p)$

\item $x_n\notin M\Rightarrow \neg R_{\initSeg{M}{n}\cup\{x_n\}}(f(n))$

\end{itemize}
but now only for $n\leq \omega(M,f)$ and $p=\phi(M,f)$.
\end{definition}
Note that Definition \ref{def-expappr} is slightly stronger than the n.c.i. (\ref{eqn-snf}), since it works for all $n\leq \omega(M,f)$ and not just $n=\omega(M,f)$.

Approximate maximal objects are useful because when a proof of a pure existential statement relies on the existence of some maximal $M$, we are typically able to find functionals $(\omega,\phi)$ which calibrate exactly how this maximal object is used, and thereby construct a witness to the existential statement in terms of an approximate maximal object relative to $(\omega,\phi)$. We will see an example of this in Section \ref{sec-case}.

\subsection{The algorithm}
\label{sec-comp-alg}

We now describe our algorithm which computes approximate maximal objects, as an intuitive state based computation $\{s_i\}_{i\in\NN}$. Here, each $s_i$ is a \emph{state}, which in this paper is defined to be a function of type $\NN\to\{(\ast)\}+\NN$ i.e. an array $s$, whose $n$th entry $s(n)$ is either a natural number or some default value $(\ast)$. The idea is that any given state encodes a current approximation to an explicit maximal object: For each state we define the set $M[s_i]\subseteq X$ as
\begin{equation*}
M[s_i]:=\{x_n\in\NN\; | \; s_i(n)=(\ast)\}
\end{equation*}
and the function $f[s_i]:\dom(X\backslash M[s_i])\to\NN$ by
\begin{equation*}
f[s_i](n):=s_i(n)\in\NN
\end{equation*}
where $s_i(n)\in\NN$ follows from the assumption that $n\notin M[s_i]$. Fixing functionals $(\omega,\phi)$, we assume for convenience that these now take as input states, and write e.g. $\omega(s_i)$ for $\omega(M[s_i],f[s_i])$. Define
\begin{equation*}
(n_i,p_i):=(\omega,\phi)(s_i).
\end{equation*}
We now describe how our state evolves. As an initial state, we set
\begin{equation*}
s_0:=\lambda n.(\ast)
\end{equation*}
and so $M[s_0]=X$ and $f[s_0]$ has an empty domain. Now, given that we are in the $i$th state, we carry out the following steps:
\begin{itemize}
\setlength\itemsep{1em}

\item Search from $0$ up to $n_i$ until some $0\leq n\leq n_i$ is found such that each of the following hold
\begin{itemize}
\setlength\itemsep{0em}

\item $x_n\in M[s_i]$,

\item $\neg R_{\initSeg{M[s_i]}{n}\cup\{x_n\}}(p_i)$

\end{itemize}

\item If no such $n$ is found, the algorithm terminates in state $s_i$.

\item Otherwise, define
\begin{equation*}
s_{i+1}:=\initSeg{s_i}{n}::p_i::\lambda k.(\ast)
\end{equation*}
(where $::$ denotes list concatenation) and so in particular, $M[s_{i+1}]=\initSeg{M[s_i]}{n}\cup \{x_k\in\NN \; | \; k>n\}$ and $x_n\notin M[s_{i+1}]$.

\end{itemize}

\begin{lemma}
\label{lem-domain}
For all states $s_i\in\NN$ and $n\in\NN$ we have
\begin{equation*}
x_n\notin M[s_i]\Rightarrow \neg R_{\initSeg{M[s_i]}{n}\cup\{x_n\}}(f[s_i](n)).
\end{equation*}
\end{lemma}

\begin{proof}
Induction on $i$. For $i=0$ the statement is trivially true, since $M[s_0]=X$. So suppose the statement is true for some $i$, and that $x_n\notin M[s_{i+1}]$. Because $M[s_{i+1}]=\initSeg{M[s_i]}{n'}\cup\{x_k\in\NN\; | \; k>n'\}$ for some $n'\leq n_i$ there are two possibilities: either $n<n'$ and $x_n\notin M[s_i]$ and so the result follow by the induction hypothesis since $f[s_{i+1}](n)=s_{i+1}(n)=s_i(n)=f[s_i](n)$ and $\initSeg{M[s_{i+1}]}{n}=\initSeg{M[s_i]}{n}$, or $n=n'$ and so $f[s_{i+1}(n)]=p_i$ which is defined to satisfy $\neg R_{\initSeg{M[s_i]}{n}\cup\{x_n\}}(p_i)$, and thus the result follows since $\initSeg{M[s_{i+1}]}{n}=\initSeg{M[s_i]}{n}$.
\end{proof}

\begin{theorem}
\label{thm-verify}
Suppose that the algorithm terminates in state $s_j$. Then $s_j$ forms an approximate explicit maximal object w.r.t. $\rhd,Q$ and $(\omega,\phi)$.
\end{theorem}

\begin{proof}
If the algorithm terminates, then by definition it holds that for all $n\leq n_j=\omega(s_j)$, if $x_n\in M[s_j]$ then $R_{\initSeg{M[s_j]}{n}\cup\{x_n\}}(p_j)$ where $p_j=\phi(s_j)$. But if $x_n\notin M[s_j]$ then $\neg R_{\initSeg{M[s_j]}{n}\cup\{x_n\}}(f[s_j](n))$ by Lemma \ref{lem-domain}, and so we're done.
\end{proof}

\subsection{Termination}

It remains, then, to show that our algorithm actually terminates on some reasonable set of parameters! Here, we make an additional abd completely standard assumption, namely that the functionals $(\omega,\phi)$ are \emph{continuous}.

\begin{definition}
\label{def-cont}
We say that $(\omega,\phi)$ are continuous if for all states $s:\NN\to\{\ast\}+\NN$ (which encode $M,f$) there exists some $L$ such that for any other input state $s'$, if $\initSeg{s}{L}=\initSeg{s'}{L}$ then 
\begin{equation*}
(\omega,\phi)(s)=(\omega,\phi)(s').
\end{equation*}
\end{definition}

Note that whenever $(\omega,\phi)$ are instantiated by computable functionals, they will automatically be continuous, so restricting ourselves to the continuous setting is entirely reasonable.

\begin{theorem}
\label{thm-termination}
Whenever the algorithm is run on continuous parameters $(\omega,\phi)$, it terminates after a finite number of steps.
\end{theorem}

\begin{proof}
Suppose that the algorithm does not terminate and thus results in an infinite run $\{s_i\}_{i\in\NN}$. We define a sequence $j_0\leq j_1\leq j_2\leq \ldots$ satisfying
\begin{equation}
\label{eqn-stability}
(\forall i\geq j_n)(\initSeg{s_i}{n}=\initSeg{s_{j_n}}{n})
\end{equation}
inductively as follows: We let $j_0:=0$, and if $j_n$ has been defined, either there exists some $j\geq j_n$ such that $x_{n}\notin M[s_j]$, in which case we define $j_{n+1}=j$, or $x_{n}\in M[s_j]$ for all $j\geq j_n$ and we set $j_{n+1}:=j_n$. To see that this construction satisfies (\ref{eqn-stability}) we use induction on $n$. The base case is trivial, so let's fix some $n$. By the induction hypothesis and the fact that $j_{n+1}\geq j_n$ we have $\initSeg{s_i}{n}=\initSeg{s_{j_{n+1}}}{n}$ for all $i\geq j_{n+1}$, and so we only need to check point $n$. Now, in the case $x_n\in M[s_i]$ for all $i\geq j_n=j_{n+1}$ we're done since this means that $s_i(n)=(\ast)$ for all $i\geq j_{n+1}$. In the other case, if $x_n\notin M[s_{j_{n+1}}]$ then $s_{j_{n+1}}(n)=p\in\NN$ and observing the manner in which the states evolves at each step, the only way this can change is if $x_m$ is removed from to $s_i$ for some $i\geq j_{n+1}$ and $m<n$. But this contradicts the induction hypothesis.

Define $s_\infty$ to be the limit of the $\initSeg{s_{j_n}}{n}$, and let $L$ be a point of continuity for $(\omega,\phi)$ on this input. Define
\begin{equation*}
j:= j_N\mbox{ \ \ \ for \ \ \ }N:=\max\{L,\omega(s_\infty)+1\}
\end{equation*}
Then in particular, since $\initSeg{s_\infty}{L}=\initSeg{s_{j}}{L}$ we must have
\begin{equation*}
n_j:=\omega(s_j)=\omega(s_\infty)<N.
\end{equation*}
But since the algorithm does not terminate, there is some $0\leq n\leq n_j$ with $x_n\in M[s_j]$ but $x_n\notin M[s_{j+1}]$. But by definition of $j=j_N$, since $n<N$ then $x_n\in M[s_j]$ implies that $x_n\in M[s_i]$ for all $i\geq j$, a contradiction. 
\end{proof}

\section{Case study: The nilradical as the intersection of all prime ideals}
\label{sec-case}

We now use our algorithm to carry out a computational analysis of the following well known fact \cite[Proposition 1.8]{AtiyahMac(1969.0)}, 
which is a frequently used form of Krull's lemma. Recall that a ring element $r$ is nilpotent if $r^e=0$ for some integer $e>0$. 

\begin{theorem}
\label{thm-nilradical}
Let $X$ be a countable commutative ring. Suppose that $r$ lies in the intersection of all prime ideals of $X$. Then $r$ is nilpotent.
\end{theorem}
We first show how the standard proof follows from our general maximality principle Theorem\,\ref{thm-zorn}.
\begin{proof}
Define $\rhd$ as in Corollary \ref{cor-ideal}, but now let $Q(x):=(\forall e)(e>0\Rightarrow x\neq r^e)$. 
Then $S\subseteq X$ is closed w.r.t $\rhd$ and satisfies $Q(S)$ iff it is an ideal which does not contain $r^e$ for any $e>0$. 
Suppose for contradiction that $r$ is not nilpotent, which would mean that $Q(\{0\})$ and thus $Q(\closure{\emptyset})$ hold. 
By Theorem \ref{thm-zorn} there is some $M$ which is maximal w.r.t. $\rhd$ and $Q$, and in this case $M\oplus x=\closure{M\cup\{x\}}$ is just 
the ideal generated by $M$ and $x$.

Take $x,y\notin M$. Then $\neg Q(M\oplus x)$ and hence there exists some $e_1>0$ such that $r^{e_1}\in M\oplus x$. 
Similarly, there exists some $e_2>0$ with $r^{e_2}\in M\oplus y$. But then $r^{e_1+e_2}\in M\oplus xy$ and thus $xy\notin M$. 
This would mean that $M$ is prime, but then $Q(M)$ contradicts the assumption that $r\in M$.
\end{proof}
\begin{lemma}
\label{lem-forall}
For $\rhd$ and $Q$ defined as in Theorem \ref{thm-nilradical}, we have
\begin{equation*}
Q(\closure{A})\Leftrightarrow (\forall b\in X^\ast,e)(\underbrace{|b|=k\wedge e>0\Rightarrow a_1\cdot b_1+\ldots+a_k\cdot b_k\neq r^e}_{R_A(b,e)})
\end{equation*}
where $A:=\{a_1,\ldots,a_k\}$, $X^\ast$ as usual denotes the set of lists over $X$ and $|b|$ is the length of $b$.
\end{lemma}

%
%


%
Our aim will be to address the following computational challenge, given any fixed $X$ and $r$,
\begin{itemize}

\item \textbf{Input.} Evidence that $r$ lies in the intersection of all prime ideals

\item \textbf{Output.} An exponent $e>0$ such that $r^e=0$

\end{itemize} 
The first question is what we take to be evidence that $r$ lies in all prime ideals. Note that this assumption is logically equivalent to the statement
\begin{equation*}
(\forall S\subseteq X)(\mbox{$S$ is not prime}\vee r\in S),
\end{equation*}
so for a computational interpretation of the above it would be reasonable to ask for a procedure which takes some $S\subseteq X$ as input, and either confirms that $r\in S$ or demonstrates that $S$ is not a prime ideal.

Let's now fix some enumeration of $X$, where we assume for convenience that $x_0=0_X$, $x_1=1_X$ and $x_2=r$. From now on we assume that we have some function 
\begin{equation*}
\psi:\mathcal{P}(X)\to \{0,1,2\}+(\{3,4,5\}\times\NN^3) 
\end{equation*}
which for any $S\subseteq X$ satisfies

\begin{itemize}

\item $\psi(S)=0\Rightarrow 0_X\notin S$

\item $\psi(S)=1\Rightarrow 1_X\in S$

\item $\psi(S)=2\Rightarrow r\in S$

\item $\psi(S)=(3,i,j,k)\Rightarrow (x_i+x_j=x_k)\wedge (x_i,x_j\in S)\wedge (x_k\notin S)$

\item $\psi(S)=(4,i,j,k)\Rightarrow (x_i\cdot x_j=x_k)\wedge (x_i\in S)\wedge (x_k\notin S)$

\item $\psi(S)=(5,i,j,k)\Rightarrow (x_i\cdot x_j=x_k)\wedge (x_i,x_j\notin S)\wedge (x_k\in S)$

\end{itemize}
The functional $\psi$ witnesses the statement that $r\in S$ or $S$ is not a prime ideal. 

\begin{lemma}
\label{lem-nilradicalmain}
Suppose that $M\subseteq X$ and $f$ satisfy
\begin{equation}
\label{eqn-nilradicalmain}
x_n\notin M\Rightarrow \neg R_{\initSeg{M}{n}\cup\{x_n\}}(f_1(n),f_2(n))
\end{equation}
where $R_A(b,e)$ is as in Lemma \ref{lem-forall} and if $f(n)=\pair{b,e}$ then $f_1(n)=b$ and $f_2(n)=e$. Whenever $\psi(M)\neq 0$ there exists some nonempty $A=\{a_1,\ldots,a_l\}\subseteq M$ together with a sequence $[b_1,\ldots,b_l]$ of elements of $X$ and $e>0$ such that
\begin{equation*}
a_1\cdot b_1+\ldots +a_l\cdot b_l=r^e.
\end{equation*}
Moreover, $e,A$ and $b$ are computable in $\psi$, $M$ and $f$.
\end{lemma}

\begin{proof}
This fairly routine case analysis is included in the appendix.
\end{proof}

\begin{lemma}
\label{lem-findcont}
Suppose that $M$ and $f$ satisfy (\ref{eqn-nilradicalmain}) as in Lemma \ref{lem-nilradicalmain} and that $\psi(M)\neq 0$. Then there exists some $n\in\NN$, sequence $b$ and $e>0$ such that 
\begin{itemize}

\item $x_n\in M$,

\item $\neg R_{\initSeg{M}{n}\cup\{x_n\}}(b,e)$

\end{itemize}
and moreover, $n$, $b$ and $e$ are computable in $\psi$, $M$ and $f$.
\end{lemma}

\begin{proof}
By Lemma \ref{lem-nilradicalmain} there exist, computable in $\psi$, $M$ and $f$, a nonempty $A=\{a_1,\ldots,a_l\}\subseteq M$ together with $b=[b_1,\ldots,b_l]$ and $e>0$ satisfying $a_1\cdot b_1+\ldots +a_l\cdot b_l=r^e$. In particular, we can find some $n\in\NN$ which is the maximal with $x_n\in A\subseteq M$, and thus $A\subseteq \initSeg{M}{n}\cup\{x_n\}$. But by expanding $b$ to some sequence $b'$ with zeroes added wherever needed, we have
\begin{equation*}
x_{\alpha_1}\cdot b'_1+\ldots +x_{\alpha_p}\cdot b'_p+x_n\cdot b'_{p+1}=r^e
\end{equation*}
where $\{x_{\alpha_1},\ldots,x_{\alpha_p}\}=\initSeg{M}{n}$, and thus $\neg R_{\initSeg{M}{n}\cup\{x_n\}}(b',e)$ holds.
\end{proof}

\begin{theorem}
\label{thm-nilalgorithm}
Given our input realizer $\psi$, define the functionals $\omega,\phi$ by
\begin{equation*}
(\omega,\phi)(M,f):=\begin{cases}n,\pair{b,e} & \mbox{if $\psi(M)\neq 0$, where $n,b$ and $e$ satisfy Lemma \ref{lem-findcont}}\\ 0,\pair{[],0} & \mbox{otherwise} \end{cases}
\end{equation*}
Suppose that the algorithm $\{s_i\}_{i\in\NN}$ described in Section \ref{sec-comp-alg} is run on $(\omega,\phi)$, and for $R_A(b,e)$ as defined in Lemma \ref{lem-forall}. Then the algorithm terminates in some final state $s_j$ satisfying
\begin{equation*}
s_j(0)_2>0\wedge r^{s_j(0)_2}=0_X.
\end{equation*}
\end{theorem}

\begin{proof}
First of all, we note that $(\omega,\phi)$ are computable, and so in particular must be continuous in the sense of Definition \ref{def-cont}. Therefore the algorithm terminates in some final state $s_j$. By Lemma \ref{lem-domain} we have
\begin{equation}
\label{eqn-nilalgorithm}
x_n\notin M[s_j]\Rightarrow \neg R_{\initSeg{M[s_j]}{n}\cup\{x_n\}}(f_1[s_j](n),f_2[s_j](n)).
\end{equation}
We claim that $\psi(M[s_j])=0$. If this were not the case, then by Lemma \ref{lem-findcont} and the definition of $(\omega,\phi)$ we would have $x_{n_j}\in M[s_j]$ and $\neg R_{\initSeg{M[s_j]}{n_j}\cup\{x_{n_j}\}}(b_j,e_j)$ for
\begin{equation*}
(n_j,\pair{b_j,e_j})=(\omega,\phi)s_j
\end{equation*}
and so by definition the algorithm cannot be in a final state. This proves the claim. But $\psi(M[s_j])=0$ implies that $x_0=0_X\notin M[s_j]$, and therefore by (\ref{eqn-nilalgorithm}) we have $\neg R_{\{x_0\}}(b,e)$ where $\pair{b,e}=f[s_j](0)=s_j(0)$, which is just
\begin{equation*}
|b|=1\wedge e>0\wedge x_0\cdot b_0=r^e.
\end{equation*}
But since $x_0\cdot b_0=0_X\cdot b_0=0$ we have $r^e=0$ i.e. $r^{s_j(0)_2}=0_X$.
\end{proof}

\subsection{Informal description of algorithm}

The basic idea behind the algorithm in this section is the following.

\begin{itemize}

\item Each state $s_i$ encodes some $M[s_i]\subseteq X$, where $x_n\notin M[s_i]$ only if we have found evidence that $\initSeg{M[s_i]}{n}\cup\{x_n\}$ generates $r^e$ for some $e>0$, in which case this evidence is encoded as $s_i(n)\in\NN$.

\item We start off at $s_0$ with the full set $M[s_0]=X$.

\item At state $s_i$ we interact with our functional $\psi$, which provides us with evidence that either $M[s_i]$ is not a prime ideal, or $r\in M[s_i]$.

\item If this evidence takes the form of anything other than $0_X\notin S$, then we are able to use this to find some $x_n\in M$ and evidence that $\initSeg{M}{n}\cup\{x_n\}$ generates $r^e$ for some $e>0$. We exclude $x_n$ from $M[s_i]$ but add all $x_k$ for all $k>n$ (since now the evidence that $\initSeg{M[s_i]}{k}\cup\{x_k\}$ generates $r^{e'}$ could be falsified by the removal of $x_n$).

\item Eventually, using a continuity argument, the algorithm terminates in some state $s_j$. The only way this can be is if $\psi(M[s_j])=0$, which indicates that $0_X\notin M[s_j]$. Thus $\{0_X\}$ generates $r^e$ for some $e>0$ encoded in the state. 

\end{itemize}

%

%








\subsection{Example: Nilpotent coefficients of invertible polynomials}

We conclude by outlining a simple and very concrete application \cite[pp.~10--11]
{AtiyahMac(1969.0)}  
of Theorem \ref{thm-nilradical}, and sketching how our algorithm would be implemented in this case. Fixing our countable commutative ring $X$, 
let $f=\sum_{i=0}^n a_i T_i$ be a unit in the polynomial ring $X[T]$. Then each $a_i$ for $i>0$ is nilpotent. 
To prove this, by Theorem \ref{thm-nilradical} it suffices to show that $a_i\in P$ for all prime ideals $P$ of $X$. 

Let $g\in X[T]$ be such that $fg=1$, and let $P$ be some arbitrary prime ideal. Then we also have $fg=1$ in $(X/P)[T]$, but since $P$ 
is prime, $X/P$ is an integral domain, and thus $0=\degr(fg)=\degr(f)+\degr(g)$. This implies that $\degr(f)=0$ in $(X/P)[T]$ and thus $a_i\in P$ for all $i>0$.

In order to obtain a concrete algorithm, which for any $a_i$ for $i>0$, produces some $e>0$ such that $r^e=0$, we need to analyse the above argument to produce a specific functional $\psi$ which for any $S\subseteq X$, witnesses the statement that either $a_i\in S$ or $S$ is not a prime ideal. Fixing $i>0$ and $S$, we define $\psi(S)$ via the following algorithm:
\begin{itemize}

\item Check in turn whether any of $0\notin S$, $1\in S$ or $a_i\in S$ are true. In the first case, return $\psi(S)=0$, and in the other, $\psi(S)=1$ and $\psi(S)=2$ respectively.

\item Otherwise, let $g=\sum_{j=0}^m b_j T^j\in X[T]$ be such that 
$1=fg=\sum^{n+m}_{k=0} c_k T^k$ for $c_k=\sum_{j=0}^k a_jb_{k-j}$. Then in particular, for $i>0$ we have 
$0=c_i=
\sum_{j=0}^{i-1} a_j b_{i-j}+a_ib_0$ 
and so (using that $a_0b_0=c_0=1$):
\begin{equation}
\label{eqn-polysum}
a_i=-a_0\sum_{j=0}^{i-1} a_j b_{i-j}.
\end{equation}
\begin{itemize}

\item Either $b_1,\ldots,b_i\in S$, and since $a_i\notin S$, an analysis of the r.h.s. of (\ref{eqn-polysum}) allows us to find, in a finite number of steps, either some $x_u,x_v\in S$ and $x_w\notin S$ such that $x_w=x_u+x_v$, in which case we return $\psi(S)=(3,u,v,w)$, or some $x_u\in S,x_v$ and $x_w\notin S$ such that $x_w=x_ux_v$, in which case we return $\psi(S)=(4,u,v,w)$.

\item Or $b_j\notin S$ for some $1\leq j\leq i$. Take $1\leq k\leq n$ and $1\leq l\leq m$ to be the maximal such that $a_k,b_l\notin S$ and consider
\begin{equation*}
0=c_{k+l}=
a_kb_l+\sum_{p+q=k+l\wedge (p>k\vee q>l)}a_pb_q. 
\end{equation*}
Then either $x_w=a_kb_l\in S$, in which case return $\psi(S)=(5,k,u,v)$ for $x_u,x_v=a_k,b_l$ or $\sum a_pb_q=-a_kb_l\notin S$, and since for each summand $a_pb_q$ either $a_p\in S$ or $b_q\in S$, an analysis identical to the previous case returns $\psi(S)=(3,u,v,w)\mbox{ or }(4,u,v,w)$ for suitable $u,v,w$.

\end{itemize}

\end{itemize}

Therefore, running our algorithm for $\psi$ as defined above results in a sequential algorithm which, by Theorem \ref{thm-nilalgorithm} terminates in some final state $s_j$ with $f[s_j]=\pair{b,e}$ for $e>0$ and $a_i^e=0$.

\begin{example}
In the very simple case where $X=\ZZ_4$ and $f=1+2T$, the corresponding run our algorithm for $a_1=2$ would be as follows;
\begin{itemize}

\item $M[s_0]=\ZZ_4$ and $\psi(\ZZ_4)=1$ (since $1\in \ZZ_4$). Remove $1$ with evidence $1\cdot 2=2^1$.

\item $M[s_1]=\ZZ_4\backslash \{1\}$ and $\psi(\ZZ_4\backslash \{1\})=2$ (since $a_1=2\in \ZZ_4\backslash\{1\})$. Remove $2$ with evidence $2\cdot 1=2^1$.

\item $M[s_2]=\ZZ_4\backslash \{1,2\}$. Noting that $(1+2T)(1+2T)=1$, we have $b_1=2\notin \ZZ_4\backslash \{1,2\}$, and so $a_1,b_1=2$ are the maximal with $a_1,b_1\notin \ZZ_4\backslash \{1,2\}$. Then $0=c_2=a_1\cdot b_1\in \ZZ_4\backslash \{1,2\}$, and thus $\psi(\ZZ_4\backslash \{1,2\})=(5,2,2,0)$, and so we remove $0$ with evidence $0+2\cdot 2=2^2$.

\item $M[s_3]=\ZZ_4\backslash \{0,1,2\}$ and $\psi(\ZZ_4\backslash \{0,1,2\})=0$, so the algorithm terminates with $e=2$.

\end{itemize}

\end{example} 
%
%
%
%

\newpage

\appendix

\section{Appendix}


\begin{proof}
[Lemma \ref{lem-nilradicalmain}]
We deal with each case in turn. Since $\psi(M)\neq 0$ there are five remaining possibilities:
\begin{itemize}
\setlength\itemsep{1em}

\item $\psi(M)=1$, i.e. $x_1=1_X\in M$ and so we set $e:=1$, $A:=\{x_1\}$ and $b:=[x_2]$ (recall that $x_2=r$).

\item $\psi(M)=2$, i.e. $x_2=r\in M$ and so $e:=1$, $A:=\{x_2\}$ and $b:=[x_1]$ work.

\item $\psi(M)=(3,i,j,k)$. Since $x_k\notin M$, by (\ref{eqn-nilradicalmain}) for $b'=f_1(k)$ we have
\begin{equation*}
x_{\alpha_1}\cdot b'_1+\ldots+x_{\alpha_p}\cdot b'_p+x_k\cdot b'_{p+1}=r^{f_2(k)}
\end{equation*}
for $\{x_{\alpha_1},\ldots,x_{\alpha_p}\}=\initSeg{M}{k}$. But then
\begin{equation*}
x_{\alpha_1}\cdot b'_1+\ldots+x_{\alpha_p}\cdot b'_p+(x_i+x_j)\cdot b'_{p+1}=r^{f_2(k)}
\end{equation*}
and so $e:=f_2(k)$, together with $A:=\{x_{\alpha_1},\ldots,x_{\alpha_p},x_i,x_j\}\subseteq M$ and $b:=[b'_1,\ldots,b'_p,b'_{p+1},b'_{p+1}]$ work.

\item $\psi(M)=(4,i,j,k)$. Entirely analogously, but this time we have
\begin{equation*}
x_{\alpha_1}\cdot b'_1+\ldots+x_{\alpha_p}\cdot b'_p+x_i\cdot (x_j\cdot b'_{p+1})=r^{f_2(k)}
\end{equation*}
and so $e:=f_2(k)$, $A:=\{x_{\alpha_1},\ldots,x_{\alpha_p},x_i\}$ and $b:=[b'_1,\ldots,b'_p,x_j\cdot b'_{p+1}]$ work.

\item $\psi(M)=(5,i,j,k)$. For $b'=f_1(i)$ and $b''=f_1(j)$ we have $x_{\alpha_1}\cdot b'_1+\ldots+x_{\alpha_p}\cdot b'_p+x_i\cdot b'_{p+1}=r^{f_2(i)}$ and $x_{\beta_1}\cdot b''_1+\ldots+x_{\beta_q}\cdot b''_q+x_j\cdot b''_{q+1}=r^{f_2(j)}$ where $\{x_{\alpha_1},\ldots,x_{\alpha_p}\}=\initSeg{M}{i}$ and $\{x_{\beta_1},\ldots,x_{\beta_q}\}=\initSeg{M}{j}$, and therefore
\begin{equation*}
\begin{aligned}
&(x_{\alpha_1}\cdot b'_1+\ldots+x_{\alpha_p}\cdot b'_p)\cdot r^{f_2(j)}+x_i\cdot b'_{p+1}\cdot (x_{\beta_1}\cdot b''_1+\ldots+x_{\beta_q}\cdot b''_{q})\\
&+x_i\cdot x_j\cdot b'_{p+1}\cdot b''_{q+1}=r^{f_2(i)+f_2(j)}
\end{aligned}
\end{equation*}
and so $e:=f_1(i)+f_2(j)$, $A:=\{x_{\alpha_1},\ldots,x_{\alpha_p},x_{\beta_1},\ldots,x_{\beta_q},x_i\cdot x_j\}$ and the corresponding $b$ from the above equation work.

\end{itemize}
\end{proof}

\bibliographystyle{splncs04}
\bibliography{tp,synergy}

\end{document}